\documentclass[twocolumn,showpacs]{revtex4}
\usepackage{graphicx}
\usepackage{tikz-cd}
\usepackage{dcolumn}
\usepackage{amsmath}
\usepackage{amssymb}
\usepackage{hyperref}
\newtheorem{theorem}{Theorem}[section]

\newtheorem{definition}[theorem]{Definition}
\newenvironment{proof}[1][Proof]{\begin{trivlist}
\item[\hskip \labelsep {\bfseries #1}]}{\end{trivlist}}

\newcommand{\qed}{\nobreak \ifvmode \relax \else
	\ifdim\lastskip<1.5em \hskip- \lastskip
	\hskip 1.5em plus0em minus0.5em \fi \nobreak
	\vrule height0.75em width0.5em depth0.25em\fi}

\begin{document}

\title{Properties of Expansion-free Dynamical Stars}

\author{Abbas \surname{Sherif}}
\email{abbasmsherif25@gmail.com}
\affiliation{Astrophysics and Cosmology Research Unit, School of Mathematics, Statistics and Computer Science, University of KwaZulu-Natal, Private Bag X54001, Durban 4000, South Africa}

\author{Rituparno \surname{Goswami}}
\email{vitasta9@gmail.com}
\affiliation{Astrophysics and Cosmology Research Unit, School of Mathematics, Statistics and Computer Science, University of KwaZulu-Natal, Private Bag X54001, Durban 4000, South Africa}

\author{Sunil \surname{Maharaj}}
\email{maharaj@ukzn.ac.za}
\affiliation{Astrophysics and Cosmology Research Unit, School of Mathematics, Statistics and Computer Science, University of KwaZulu-Natal, Private Bag X54001, Durban 4000, South Africa}

\begin{abstract}
We study the geometrical and dynamical features of expansion-free dynamical stars in general relativity. Such stars can exist only if particular physical and geometric conditions are satisfied. Firstly, for trapping to exist in an expansion-free dynamical star, the star must accelerate and radiate simultaneously. If either are zero then the shear \(\Sigma\) must be zero through out the star, in which case the star is static (\(\Theta=\Sigma=0\)). Secondly, we prove that with nonzero acceleration and radiation, expansion-free dynamical stars must be conformally flat.
\end{abstract}
\maketitle

\section{Introduction}

Models of radiating stars in general relativity are important to describe astrophysical processes and to study gravitational collapse. Some recent examples of exact models which are physically reasonable were obtained by Tewari and Charan \cite{tac1}, Tewari \cite{tac2}, and Ivanov \cite{iv1,iv2,iv3}. Anisotrophy and dissipative effects have been shown to influence the collapse rate and temperature profiles in radiating stars by Reddy \textit{et al} \cite{red1}. It has been demonstrated that classes of exact solutions exist in general relativity, referred to as Euclidean stars, which regain Newtonian stars in the appropriate limit \cite{sant1,gov1,gov2}. The Lie analysis of differential equations using symmetry invariance has proved to be a systematic method to produce general categories of exact solutions to the boundary condition of radiating objects \cite{gov3,moh1,moh2}. An important class of radiating stars which are expansion-free was introduced by Herrera \textit{et al} \cite{lh1}. Expansion-free dynamical models implies the existence of a cavity or void. Matter distributions with a vanishing expansion scalar have to be inhomogeneous. These physical features should have important astrophysical consequences for spherically symmetric distributions. Studies containing the description of physical properties of expansion-free dynamical radiating stars are contained in several treatments \cite{lh5,kum1,kum2}. Therefore it is important to study the geometrical properties of expansion-free dynamical stars and find general conditions for their existence.

The aim of this paper is to investigate under what conditions can there be trapping in a relativistic expansion-free dynamical star. This analysis falls in the scope of stability analysis of self-gravitating systems (some of the references are given in \cite{mr1,jw1,nos1,gmm1,rittt1,rittt2}). We will consider the conditions on the acceleration and radiation quantities that allow for trapping in such stars. It is also an interesting exercise, with all the different quantities acting on such stars, to determine the geometry as these structures evolve. We will make use of the equivalent forms of the field equations from the \(1+1+2\) semi-tetrad covariant formulation of general relativity \cite{pg1,ts1,ts2,ts3,cc1,rit1}. The semi-tetrad formalism has been a useful approach in displaying geometrical features in a transparent fashion which are difficult to find using other approaches.

Various authors have explored expansion-free dynamical models with different considerations. The central theme of the interest in such models is the possibility that they could help explain the existence of voids on cosmological scales. In 2008, Herrera and co-authors \cite{lh1} studied such models with non-zero shear and showed that the appearance of a cavity (see reference \cite{pp1} for more discussion), with matter which is anisotropic and dissipative, undergoing explosion is inevitable. The same authors followed this result by a 2009 paper in \cite{lh2} in which they ruled out the Skripkin expansion-free dynamcal model (see reference \cite{skr1}) with constant energy density and isotropic pressure. Another study in \cite{lh3} involved the study of models collapsing adiabatically, and showed that the instability was independent of the star's stiffness. In particular, it was shown that the instability is entirely governed by the pressures and the radial profile of the energy density. 

In section \ref{02} we give a short overview of the \(1+1+2\) semi-tetrad formulation. In section \ref{03} we present the results of the paper. We conclude with a discussion of the results in section \ref{04}.

\section{Preliminaries}\label{02}

We provide some background material in this section, covering the \(1+1+2\) semi-tetrad covariant formalism as well as notes on and calculations of useful quantities, utilized in this paper. 

We start by explicitly defining \textit{locally rotationally symmetric} class II spacetimes \cite{ggff1,sge1}.
\begin{definition}
A \textbf{locally rotationally symmetric class} II (LRS II) spacetime is an evolving and vorticity free (zero rotation) and spatial twist free spacetime with a one dimensional isotropy group of spatial rotations defined at each point of the spacetime. It is given by the general line element \begin{eqnarray}\label{fork}
ds^2&=&-A^2\left(t,\chi\right)+ B^2\left(t,\chi\right)\notag \\
&&+ C^2\left(t,\chi\right) \left(dy^2+D^2\left(y,k\right)dz^2\right),
\end{eqnarray}
where \(t,\chi\) are parameters along integral curves of the timelike vector field \(u^a=A^{-1}\delta^a_0\) of a timelike congruence and the preferred spacelike vector \(e^a=B^{-1}\delta_{\nu}^a\) respectively. The constant \(k\) fixes the function \(D\left(y,k\right)\) (\(k=-1\) corresponds to \(\sinh y\), \(k=0\) corresponds to \(y\), \(k=1\) corresponds to \(\sin y\)) \cite{sge1}. 
\end{definition}
LRS II spacetimes generalize spherically symmetric spacetimes, and can be used to study astrophysical bodies such as stars and their evolution. From the line element in \eqref{fork} it is clear that most physically realistic and interesting spacetimes fall within the LRS II class. 

Let us next introduce the \(1+1+2\) covariant splitting of spacetime and the resulting fields equations for LRS II spacetimes \cite{cc1,rit1}. 

To start with, let (\(M,g_{ab}\)) be a spacetime manifold. To any timelike congruence we associate a unit vector field \(u^a\) tangent to the congruence for which \(u^au_a=-1\). Given any \(4\)-vector \(U^a\) in the spacetime, the projection tensor \(h_a^{\ b}\equiv g_a^{\ b}+u_au^b\), projects \(U^a\) onto the \(3\)-space as
\begin{eqnarray*}
U^a&=&Uu^a + U^{\langle a \rangle },
\end{eqnarray*}
where \(U\) is the scalar along \(u^a\) and \(U^{\langle a \rangle }\) is the projected \(3\)-vector \cite{ggff2}. This naturally gives rise to two derivatives:
\begin{itemize}
\item The \textit{covariant time derivative} (or simply the dot derivative)  along the observers' congruence. For any tensor \(S^{a..b}_{\ \ \ \ c..d}\), \(\dot{S}^{a..b}_{\ \ \ \ c..d}\equiv u^e\nabla_eS^{a..b}_{\ \ \ \ c..d}\).

\item Fully orthogonally \textit{projected covariant derivative} \(D\) with the tensor \(h_{ab}\), with the total projection on all the free indices. For any tensor \(S^{a..b}_{\ \ \ \ c..d}\), \(D_eS^{a..b}_{\ \ \ \ c..d}\equiv h^a_{\ f}h^p_{\ c}...h^b_{\ g}h^q_{\ d}h^r_{\ e}\nabla_rS^{f..g}_{\ \ \ \ p..q}\).
\end{itemize}
This \(1+3\) splitting irreducibly splits the covariant derivative of \(u^a\) as
\begin{eqnarray}\label{mmmn}
\nabla_au_b=-A_au_b+\frac{1}{3}h_{ab}\Theta+\sigma_{ab}.
\end{eqnarray}
In \eqref{mmmn}, \(A_a=\dot{u}_a\) is the acceleration vector, \(\Theta\equiv D_au^a\) is the expansion and \(\sigma_{ab}=D_{\langle b}u_{a\rangle}\) is the shear tensor (wherever used in this paper, angle brackets will denote the projected symmetric trace-free part of the tensor). LRS II spacetimes also have the property that the Weyl tensor is purely electric as the magnetically part of the Weyl tensor is identically zero (see reference \cite{cc1} for details). 

The splitting also allows for the energy momentum tensor to be decomposed as
\begin{eqnarray}
T_{ab}=\rho u_au_b + 2q_{(a}u_{b)} +ph_{ab} + \pi_{ab},
\end{eqnarray}
where \(\rho\equiv T_{ab}u^au^b\) is the energy density, \(q_a=-h_a^{\ c}T_{cd}u^d\) is the \(3\)-vector defining the heat flux, \(p\equiv\left(1/3\right)h^{ab}T_{ab}\) is the isotropic pressure and \(\pi_{ab}\) is the anisotropic stress tensor.

If there is a preferred unit normal spatial direction \(e^a\) as is the case with LRS II spacetimes, the metric \(g_{ab}\) can be split into terms along the \(u^a\) and \(e^a\) directions (the vector field \(e^a\) splits the \(3\)-space), as well as on the \(2\)-surface, i.e. 
\begin{eqnarray}
g_{ab}=N_{ab}-u_au_b+e_ae_b,
\end{eqnarray} 
where the projection tensor \(N_{ab}\) projects any two vector orthogonal to \(u^a\) and \(e^a\) onto the \(2\)-surface defined by the sheet \(N^{\ \ a}_a=2\) (\(u^aN_{ab}=0, \ e^aN_{ab}=0\)), and \(e^a\) is defined such that \(e^ae_a=1\) and \(u^ae_a=0\). This is referred to as the \(1+1+2\) splitting. This splitting of the spacetime additionally gives rise to the splitting of the covariant derivatives along the \(e^a\) direction and on the \(2\)-surface:
\begin{itemize}
\item The \textit{hat derivative} is the spatial derivative along the vector \(e^a\): for a \(3\)-tensor \(\psi_{a..b}^{\ \ \ \ c..d}\), \(\hat{\psi}_{a..b}^{\ \ \ \ c..d}\equiv e^fD_f\psi_{a..b}^{\ \ \ \ c..d}\).

\item The \textit{delta derivative} is the projected spatial derivative on the \(2\)-sheet by \(N_a^{\ b}\) and projected on all the free indices: for any \(3\)-tensor \(\psi_{a..b}^{\ \ \ \ c..d}\), \(\delta_e\psi_{a..b}^{\ \ \ \ c..d}\equiv N_a^{\ f}..N_b^{\ g}N_h^{\ c}..N_i^{\ d}N_e^{\ j}D_j\psi_{f..g}^{\ \ \ \ h..i}\).
\end{itemize} 

For LRS II spacetimes, the \(1+1+2\) covariant scalars fully describing the LRS II spacetimes are \cite{cc1}
\begin{eqnarray*}
\lbrace{A,\Theta,\phi, \Sigma, \mathcal{E}, \rho, p, \Pi, Q\rbrace}. 
\end{eqnarray*}
The quantity \(\phi\equiv\delta_ae^a\) is the sheet expansion, \(\Sigma\equiv\sigma_{ab}e^ae^b\) is the scalar associated to the shear tensor \(\sigma_{ab}\), \(\mathcal{E}\equiv E_{ab}e^ae^b\) is the scalar associated with the electric part of the Weyl tensor \(E_{ab}\), \(\Pi\equiv\pi_{ab}e^ae^b\) is the anisotropic stress scalar, \(Q\equiv -e^aT_{ab}u^b=q_ae^a\) is the scalar associated to the heat flux vector \(q_a\) .

The full covariant derivatives of the vector fields \(u^a\) and \(e^a\) are given by \cite{cc1}
\begin{eqnarray}\label{4}
\nabla_au_b&=&-Au_ae_b + e_ae_b\left(\frac{1}{3}\Theta + \Sigma\right) \notag \\
&&+ N_{ab}\left(\frac{1}{3}\Theta -\frac{1}{2}\Sigma\right),
\end{eqnarray}
\begin{eqnarray}\label{444}
\nabla_ae_b&=&-Au_au_b + \left(\frac{1}{3}\Theta + \Sigma\right)e_au_b +\frac{1}{2}\phi N_{ab}.
\end{eqnarray}
We also note the useful expression 
\begin{eqnarray}\label{redpen}
\hat{u}^a&=&\left(\frac{1}{3}\Theta+\Sigma\right)e^a.
\end{eqnarray}
When acting on a scalar \(\psi\), the dot (\(\dot{\ }\)) and hat (\(\hat{\ }\)) derivatives satisfy the commutation relation \cite{cc1}
\begin{eqnarray}\label{apr512}
\hat{\dot{\psi}} - \dot{\hat{\psi}} = -A\dot{\psi} + \left(\frac{1}{3}\Theta + \Sigma\right)\hat{\psi}.
\end{eqnarray}
This is a useful relation that will be utilized often in our calculations. 

The evolution and propagation equations may be obtained from using the Ricci identities of the vectors \(u^a\) and \(e^a\) as well as the doubly contracted Bianchi identities \cite{cc1,rit1}. The evolution and propagation equations are given as follows. (For full derivation of the equations see \cite{cc1}.)\\
\begin{itemize}
\item \textit{Evolution (LRS II):}
\begin{eqnarray}\label{evo1}
\frac{2}{3}\dot{\Theta}-\dot{\Sigma}&=&A\phi-2\left(\frac{1}{3}\Theta - \frac{1}{2}\Sigma\right)^2-\frac{1}{3}\left(\rho+3p\right)+\mathcal{E}\notag\\
&&-\frac{1}{2}\Pi,
\end{eqnarray}
\begin{eqnarray}\label{evo100}
\dot{\phi}&=&\left(\frac{2}{3}\Theta-\Sigma\right)\left(A-\frac{1}{2}\phi\right)+Q,
\end{eqnarray}
\begin{eqnarray}\label{evo101}
\dot{\mathcal{E}}-\frac{1}{3}\dot{\rho}+\frac{1}{2}\dot{\Pi}&=&-\frac{3}{2}\left(\frac{2}{3}\Theta-\Sigma\right)\mathcal{E}-\frac{1}{4}\left(\frac{2}{3}\Theta-\Sigma\right)\Pi \notag\\
&&+\frac{1}{2}\phi Q+\frac{1}{2}\left(\rho+p\right)\left(\frac{2}{3}\Theta-\Sigma\right).
\end{eqnarray}
\item \textit{Propagation (LRS II):}
\begin{eqnarray}\label{evo2}
\frac{2}{3}\hat{\Theta}-\hat{\Sigma}&=&\frac{3}{2}\phi \Sigma +Q,
\end{eqnarray}
\begin{eqnarray}\label{evo200}
\hat{\phi}&=&\left(\frac{1}{3}\Theta+\Sigma\right) \left(\frac{2}{3}\Theta-\Sigma\right)-\frac{1}{2}\phi^2 -\frac{2}{3}\rho-\mathcal{E}\notag\\
&&-\frac{1}{2}\Pi,
\end{eqnarray}
\begin{eqnarray}\label{evo201}
\hat{\mathcal{E}}-\frac{1}{3}\hat{\rho}+\frac{1}{2}\hat{\Pi}&=&-\frac{3}{2}\phi\left(\mathcal{E}+\frac{1}{2}\Pi\right)\notag\\
&&-\frac{1}{2}\left(\frac{2}{3}\Theta-\Sigma\right)Q.
\end{eqnarray}
\item \textit{Propagation/Evolution (LRS II):}
\begin{eqnarray}\label{evo3}
\hat{A}-\dot{\Theta}&=&-\left(A+\phi\right)A+\frac{1}{3}\Theta^2+\frac{3}{2}\Sigma^2\notag\\
&&+\frac{1}{2}\left(\rho+3p\right),
\end{eqnarray}
\begin{eqnarray}\label{evo300}
\hat{Q}+\dot{\rho}&=&-\Theta\left(\rho+p\right)-\left(\phi+2A\right)Q-\frac{3}{2}\Sigma\Pi,
\end{eqnarray}
\begin{eqnarray}\label{evo301}
\hat{p}+\hat{\Pi}+\dot{Q}&=&-\left(\frac{3}{2}\phi+A\right)\Pi-\left(\frac{4}{3}\Theta+\Sigma\right)Q\notag\\
&&-\left(\rho+p\right)A.
\end{eqnarray}
\end{itemize}

The outgoing null expansion, whose vanishing necessitates trapping, has been calculated in references \cite{rit1,abs1} as
\begin{eqnarray}\label{piapia1}
\Theta_k&=&\frac{1}{\sqrt{2}}\left(\frac{2}{3}\Theta-\Sigma+\phi\right).
\end{eqnarray}
The equation of the outgoing null expansion scalar here corresponds to equation \(\left(32\right)\) of \cite{rit1}, but we have unitized the energy function for our choice of the outgoing null normal vector field \(k^a\), whose divergence gives \(\Theta_k\). It is clear from \eqref{piapia1} that even with the vanishing of \(\Theta\), it is still possible to have trapping. This is the main focus of this work and is investigated in the next section.

\section{Results}\label{03}
In this section we state and prove the results of the paper.

\subsection{Dynamics of expansion-free stars}
We state and prove the following theorem.
\begin{theorem}\label{tth1}
\textit{An expansion-free dynamical star must accelerate and radiate simultaneously.}
\end{theorem}

\begin{proof}
We establish this by fixing both the acceleration and the heat flux to zero, and then by fixing either of the acceleration or the heat flux to zero.

\subsubsection{Case 1:}

First, suppose \(A=0\ \ \text{and}\ \ Q=0\). From \eqref{evo3} we have the algebraic constraint equation
\begin{eqnarray}\label{apr511}
0&=&\frac{3}{2}\Sigma^2+\frac{1}{2}\left(\rho+3p\right).
\end{eqnarray}
Here we note that since \(\Sigma^2>0\), \eqref{apr511} implies that the strong energy condition must be violated, i.e. \(\rho+3p<0\). For \(A=\Theta=0\), \eqref{apr512} is simply
\begin{eqnarray}\label{apr513}
\hat{\dot{\psi}} - \dot{\hat{\psi}} = \Sigma\hat{\psi}.
\end{eqnarray}
Taking the hat derivative of \eqref{evo100} and the dot derivative of \eqref{evo200} we obtain respectively
\begin{eqnarray}\label{hhhh1}
\hat{\dot{\phi}}&=&\frac{1}{2}\hat{\phi}\Sigma+\frac{1}{2}\phi\hat{\Sigma}\notag\\
&=&-\left(\phi^2+\frac{1}{2}\Sigma^2+\frac{1}{3}\rho+\frac{1}{2}\mathcal{E}+\frac{1}{4}\Pi\right)\Sigma
\end{eqnarray}
and
\begin{eqnarray}\label{hhhh2}
\dot{\hat{\phi}}&=&-2\Sigma\dot{\Sigma}-\phi\dot{\phi}-\frac{2}{3}\dot{\rho}-\dot{\mathcal{E}}-\frac{1}{2}\dot{\Pi}\notag\\
&=&-\left(\frac{1}{2}\phi^2+\Sigma^2+\frac{1}{6}\rho-\frac{1}{2}\mathcal{E}+\frac{3}{2}p\right)\Sigma.
\end{eqnarray}
Using the commutation relation on \eqref{hhhh1} and \eqref{hhhh2}, we obtain
\begin{eqnarray}\label{apr515}
\left[\frac{1}{4}\Pi+\frac{3}{2}\Sigma^2 + \frac{1}{2}\left(\rho+3p\right)\right]\Sigma=0.
\end{eqnarray}
So either \(\Sigma=0\) or
\begin{eqnarray*}
\frac{1}{4}\Pi+\frac{3}{2}\Sigma^2 + \frac{1}{2}\left(\rho+3p\right)=0. 
\end{eqnarray*}
If \(\Sigma=0\) then the star must be static (\(\Theta=\Sigma=0\)), so we assume that \(\Sigma\neq0\) and that 
\begin{eqnarray}\label{augie1001}
\frac{1}{4}\Pi+\frac{3}{2}\Sigma^2 + \frac{1}{2}\left(\rho+3p\right)=0.
\end{eqnarray}
From \eqref{apr511}, \eqref{augie1001} implies that \(\Pi=0\). Now if we take the dot derivative of \eqref{apr511} and substitute for \eqref{evo1} and \eqref{evo300}, we obtain the evolution of \(p\)
\begin{eqnarray}\label{apr516}
\dot{p}=\left(\Sigma^2+2\mathcal{E}\right)\Sigma.
\end{eqnarray}
Taking the hat derivative of \eqref{apr516} and the dot derivative of \eqref{evo301} we obtain respectively
\begin{eqnarray}\label{nngg1}
\hat{\dot{p}}=\left(-\frac{9}{2}\phi\Sigma^2-6\phi\mathcal{E}+\frac{2}{3}\hat{\rho}\right)\Sigma
\end{eqnarray}
and
\begin{eqnarray}\label{apr517}
\dot{\hat{p}}=0.
\end{eqnarray}
Using the commutation relation on \eqref{nngg1} and \eqref{apr517} we obtain the propagation of \(\rho\)
\begin{eqnarray}\label{apr518}
\hat{\rho}=9\phi\left(\frac{3}{4}\Sigma^2+\mathcal{E}\right), 
\end{eqnarray}
Now taking the hat derivative of \eqref{evo1} and the dot derivative of \eqref{evo2}, we obtain respectively 
\begin{eqnarray}\label{augie10090}
\hat{\dot{\Sigma}}&=&\Sigma\hat{\Sigma}-\hat{\mathcal{E}}\notag\\
&=&-\frac{3}{2}\phi\Sigma^2-\hat{\mathcal{E}}
\end{eqnarray}
and
\begin{eqnarray}\label{augie10091}
\dot{\hat{\Sigma}}&=&-\frac{3}{2}\left(\dot{\phi}\Sigma+\phi\dot{\Sigma}\right)\notag\\
&=&-\frac{3}{2}\phi\mathcal{E}.
\end{eqnarray}
Using the commutation relation on \eqref{augie10090} and \eqref{augie10091} we obtain 
\begin{eqnarray}\label{augie10092}
\hat{\mathcal{E}}=\frac{3}{2}\phi\left(\Sigma^2-\mathcal{E}\right),
\end{eqnarray}
which upon substituting in \eqref{evo201} we obtain
\begin{eqnarray}\label{augie10093}
\hat{\rho}=\frac{9}{2}\phi\Sigma^2.
\end{eqnarray}
Comparing \eqref{apr518} and \eqref{augie10093} we get 
\begin{eqnarray}\label{augie10094}
\phi\left(\mathcal{E}+\frac{1}{4}\Sigma^2\right)=0.
\end{eqnarray}
Therefore we must have either \(\phi=0\) or
\begin{eqnarray}\label{augie100100}
\mathcal{E}=-\frac{1}{4}\Sigma^2.
\end{eqnarray}
We show that either case yields \(\Sigma=0\), in which case the star is static. First, suppose \(\phi=0\). Then \eqref{evo200} gives the constraint equation
\begin{eqnarray}\label{apr519}
\Sigma^2=-\frac{2}{3}\rho-\mathcal{E}, 
\end{eqnarray}
and comparing \eqref{apr511} and \eqref{apr519} we obtain
\begin{eqnarray}\label{apr520}
\mathcal{E}=-\frac{1}{3}\left(\rho-3p\right).
\end{eqnarray}
Taking the dot derivative of \eqref{apr520}, using \eqref{evo101} and \eqref{apr516} we obtain
\begin{eqnarray}\label{apr521}
\Sigma\left[\Sigma^2+\frac{1}{2}\left(\mathcal{E}+\rho+p\right)\right]=0.
\end{eqnarray}
Again, assuming \(\Sigma\neq0\) we must have 
\begin{eqnarray}\label{aprl1}
\Sigma^2=-\frac{1}{2}\left(\mathcal{E}+\rho+p\right).
\end{eqnarray}
Taking the dot derivative of \eqref{aprl1} and using \eqref{evo1} we obtain
\begin{eqnarray}\label{aprl2}
\dot{p}=\Sigma\left(4\mathcal{E}-2\Sigma^2\right).
\end{eqnarray}
and upon comparing to \eqref{apr516} we obtain
 \begin{eqnarray}\label{aprl3}
\Sigma\left(2\mathcal{E}-3\Sigma^2\right)=0.
\end{eqnarray}
Since \(\Sigma\neq0\) we have 
 \begin{eqnarray}\label{aprl4}
\mathcal{E}=\frac{3}{2}\Sigma^2,
\end{eqnarray}
which, uupon sing \eqref{apr511} and \eqref{apr520} we obtain
\begin{eqnarray}\label{aprl5}
\rho=-\frac{5}{3}p.
\end{eqnarray}
Taking the dot derivative of \eqref{aprl5} we have \(\dot{p}=0\). Setting \eqref{apr516} to zero, while using \eqref{aprl4} to substitute for \(\mathcal{E}\) gives \(\Sigma^2=0\) which gives \(\Sigma=0\).

Next assume \(\phi\neq0\) and that \eqref{augie100100} is satisfied. Now taking the dot derivative of \eqref{augie100100}, and using \eqref{evo101}, \eqref{evo1} and \eqref{augie100100} to simplify we obtain
\begin{eqnarray}\label{augie10095}
\left[\frac{3}{4}\Sigma^2+\rho+p\right]\Sigma=0,
\end{eqnarray}
Assume \(\Sigma\neq0\). We must have
\begin{eqnarray}\label{augie100111}
\frac{3}{4}\Sigma^2+\rho+p=0,
\end{eqnarray}
Using \eqref{apr511}, \eqref{augie100111} simplifies to
\begin{eqnarray}\label{augie10096}
\rho=-\frac{14}{5}p.
\end{eqnarray}
Finally taking the dot derivative of \eqref{augie10096} we have \(\dot{p}=0\), which upon comparing to \eqref{apr516} and substituting for \(\mathcal{E}\) using \eqref{augie100100} we obtain \(\Sigma^2=0\) which gives \(\Sigma=0\).

\subsubsection{Case 2:}

Let us next consider the case \(A\neq 0\) and \(Q=0\). The commutation relation \eqref{apr512}, now becomes
\begin{equation}\label{apr5121}
\hat{\dot{\psi}} - \dot{\hat{\psi}} = -A\dot{\psi} + \Sigma\hat{\psi}.
\end{equation}
Taking the hat derivative of \eqref{evo101} and the dot derivative of \eqref{evo201}, we obtain respectively 
\begin{eqnarray}\label{augie10030}
\hat{\dot{\Sigma}}&=&-\hat{A}\phi-A\hat{\phi}+\Sigma\hat{\Sigma}+\frac{1}{3}\hat{\rho}+\hat{p}-\hat{\mathcal{E}}+\frac{1}{2}\hat{\Pi}\notag\\
&=&A^2\phi +\frac{3}{2}A\phi^2 - 3\phi\Sigma^2-\frac{1}{3}A\rho-\frac{3}{4}\phi\Pi -\frac{1}{2}\phi\rho\notag\\
&& - \frac{3}{2}\phi p + A\mathcal{E} - \frac{1}{2}A\Pi - Ap+\frac{3}{2}\phi\mathcal{E}
\end{eqnarray}
and
\begin{eqnarray}\label{augie10031}
\dot{\hat{\Sigma}}&=&-\frac{3}{2}\left(\dot{\phi}\Sigma+\phi\dot{\Sigma}\right)\notag\\
&=&\frac{3}{2}A\Sigma^2-\frac{3}{2}\phi\Sigma^2 + \frac{3}{2}A\phi^2 - \frac{1}{2}\phi\rho-\frac{3}{2}\phi p \\
&&+\frac{3}{2}\phi\mathcal{E}-\frac{3}{4}\phi\Pi.
\end{eqnarray}
Using the commutation relation on \eqref{augie10030} and \eqref{augie10031} we obtain 
\begin{eqnarray}\label{apr534}
A\Sigma^2=0.
\end{eqnarray}
Since \(A\neq 0\), we must have \(\Sigma=0\), and thus the star is static.

\subsubsection{Case 3:}

Finally, we consider the case \(A=0\) and \(Q\neq 0\). From \eqref{evo3} we have the constraint equation as \eqref{apr511}. The commutation relation in this case is \eqref{apr513}. Taking the hat derivative of \eqref{evo100} and the dot derivative of \eqref{evo200}, we obtain respectively
\begin{eqnarray}\label{augie10020}
\hat{\dot{\phi}}&=&\frac{1}{2}\hat{\Sigma}\phi+\frac{1}{2}\Sigma\hat{\phi}\notag\\
&=&-\phi^2\Sigma - \frac{1}{2}\phi Q -\frac{1}{2}\Sigma^3-\frac{1}{3}\Sigma\rho-\frac{1}{2}\Sigma\mathcal{E}-\frac{1}{4}\Sigma\Pi\notag\\
&&+\hat{Q}
\end{eqnarray}
and
\begin{eqnarray}\label{augie10021}
\dot{\hat{\phi}}&=&-2\Sigma\dot{\Sigma}-\phi\dot{\phi}-\frac{2}{3}\dot{\rho}-\left(\dot{\mathcal{E}}+\frac{1}{2}\dot{\Pi}\right)\notag\\
&=&-\Sigma^3-\frac{1}{6}\Sigma\rho - \frac{5}{4}\Sigma\Pi + \frac{1}{2}\Sigma\mathcal{E}+\frac{1}{2}\Sigma p - \dot{\rho}.
\end{eqnarray}
Using the commutation relation on \eqref{augie10020} and \eqref{augie10021} we obtain
\begin{eqnarray}\label{apr545}
\left[\frac{3}{2}\Sigma^2+\frac{1}{2}\left(\rho+3p\right)\right]\Sigma&=&\phi Q,
\end{eqnarray}
which, upon using \eqref{apr511} reduces to
\begin{eqnarray}\label{apr546}
\phi Q &=& 0.
\end{eqnarray}
Since \(Q\neq 0\), we must have \(\phi=0\). But from \eqref{evo100}, this gives \(Q=0\). 

From the three cases considered, we therefore must have \(A\neq 0\) and \(Q\neq 0\) to have an expansion-free star that is evolving.\qed
\end{proof}

\subsection{Geometry of expansion-free stars}
We state and prove the following theorem on the geometry of expansion-free dynamical stars.
\begin{theorem}\label{tth2}
\textit{An expansion-free dynamical star must be conformally flat.}
\end{theorem}

\begin{proof}
We prove this by checking for additional constraints from the field equations with \(A\neq 0\) and \(Q\neq 0\). The commutation relation in this case is given by \eqref{apr5121}. Taking the hat derivative of \eqref{evo1} we obtain 
\begin{eqnarray}\label{apr611}
\hat{\dot{\Sigma}}&=&A^2\phi + \frac{3}{2}A\phi^2 - 3\phi\Sigma^2 + A\Sigma^2 - \frac{1}{2}\phi\rho - \frac{3}{2}\phi p\notag \\ &&+ \frac{2}{3}A\rho + \frac{1}{2}A\Pi - \frac{3}{2}\Sigma Q + \frac{3}{2}\phi\mathcal{E}+ \frac{3}{4}\phi\Pi + \hat{p}\notag\\
&& + \hat{\Pi},
\end{eqnarray}
and taking the dot derivative of \eqref{evo2} we have
\begin{eqnarray}\label{apr612}
\dot{\hat{\Sigma}}&=&\frac{3}{2}A\Sigma^2 - \frac{3}{2}\phi\Sigma^2 - \frac{3}{2}\Sigma Q + \frac{3}{2}A\phi^2 - \frac{1}{2}\phi\rho \notag\\
&&- \frac{1}{2}\phi p + \frac{3}{2}\phi \mathcal{E} - \frac{3}{4}\phi\Pi - \dot{Q}.
\end{eqnarray}
Taking the difference of \eqref{apr611} and \eqref{apr612} and employing the commutation relation \eqref{apr5121}, we obtain

\begin{eqnarray}\label{apr613}
-A\dot{\Sigma} + \Sigma\hat{\Sigma}&=&A^2\phi - \frac{3}{2}\phi\Sigma^2 - \frac{1}{2}A\Sigma^2 + \frac{3}{2}\phi\Pi \notag\\
&&+ \frac{2}{3}A\rho + \frac{1}{2}A\Pi + \dot{Q} + \hat{p} + \hat{\Pi},
\end{eqnarray}
which upon using \eqref{evo301} and simplifying gives
\begin{eqnarray}\label{apr556}
A\mathcal{E}=0.
\end{eqnarray}
Since \(A\neq 0\), we must have \(\mathcal{E}=0\) so that the electric part of the Weyl tensor is vanishing. Fixing \(\mathcal{E}=0\) in the field equations for \(A\neq 0, Q\neq 0\), all other commutation relations on pairs of evolution and propagation equations return identities.\qed
\end{proof}

\section{Discussion}\label{04}

The expansion-free condition in general relativity has received considerable attention in recent years and has been applied to describe physical features of radiating stars. We have utilized the \(1+1+2\) semi-tetrad covariant formalism to study such stars in general in spherical symmetry. The analysis shows that expansion-free dynamical stars are severely constrained, and can only exist under very particular conditions. From the set of field equations, we have explicitly shown that a necessary condition for a star with zero expansion to evolve is that the star has non-zero radiation and acceleration. With further analysis of the field equations with \(A\neq 0\) and \(Q\neq 0\), it is shown that the star is necessarily conformally flat. Proving these results amount to the analysis of the field equation via commutation relations, via which we obtain additional constraints, which further give us additional evolution equations that can be matched against the original set of equations. These results add to the literature on expansion-free dynamical stars, which have been developed over the last decade and half, most notably through works of Herrera and co-authors.

\section*{Acknowledgements}
The authors would like to thank the anonymous referee for the very kind and valuable comments that have helped this paper to read clearer. AS and RG are supported by the National Research Foundation (NRF), South Africa. SDM acknowledges that this work is based on research supported by the South African Research Chair Initiative of the Department of Science and Technology and the National Research Foundation.

\end{document}